\let\accentvec\vec  
\documentclass{llncs}
\let\vec\accentvec 

\usepackage{xspace}
\usepackage{amssymb}
\usepackage{graphicx} 	
\usepackage[shortlabels]{enumitem}
\setlist[enumerate,1]{label=\arabic*., ref=(\arabic*)}

\usepackage{url}
\usepackage{amsmath}
\usepackage{multirow}

\usepackage[tight]{subfigure} 

\usepackage{float}

\usepackage{tikz}
\usetikzlibrary{shapes}
\usetikzlibrary{calc}

\usepackage{bibunits}

\usepackage[pdftex,                
            bookmarks=true,        
            bookmarksnumbered=true,
            bookmarksopen=false,
            hypertexnames=false,
]{hyperref}
\makeatletter
\providecommand*{\toclevel@title}{0}
\providecommand*{\toclevel@author}{0}
\makeatother

\usepackage{xcolor}
\definecolor{dark-red}{rgb}{0.4,0.15,0.15}
\definecolor{dark-blue}{rgb}{0.15,0.15,0.4}
\definecolor{medium-blue}{rgb}{0,0,0.5}
\definecolor{gray}{rgb}{0.5,0.5,0.5}

\hypersetup{
    colorlinks, linkcolor={dark-red},
    citecolor={dark-blue}, urlcolor={medium-blue}
}

\newenvironment{relemma}[1]{\medskip\noindent\textbf{Lemma \ref{#1}.}\it}{}
\newenvironment{retheorem}[1]{\medskip\noindent\textbf{Theorem \ref{#1}.}\it}{}

\usepackage{doi}

\newcommand{\problemdef}[3]
{
\begin{quote}
\textsc{#1}\\
\textbf{Input:} #2\\
\textbf{Question:} #3
\end{quote}
}

\newcommand{\parproblemdef}[4]
{
\begin{quote}
\textsc{#1}\\
\textbf{Input:} #2\\
\textbf{Parameter:} #3\\
\textbf{Question:} #4
\end{quote}
}

\newcommand{\ttree}{\mathbb{T}}
\newcommand{\tobs}{\widetilde{\mathbb{T}}}

\newcommand{\yes}[0]{\textsc{yes}\xspace}
\newcommand{\no}[0]{\textsc{no}\xspace}

\newcommand{\orsymb}[0]{\textsc{or}\xspace}
\newcommand{\andsymb}[0]{\textsc{and}\xspace}

\newcommand{\dotcup}{\mathop{\dot\cup}}

\newcommand{\npeqconp}[0]{NP~$=$~coNP\xspace}
\newcommand{\containment}[0]{NP~$\subseteq$ coNP$/$poly\xspace}
\newcommand{\ncontainment}[0]{NP~$\not \subseteq$ coNP$/$poly\xspace}
\newcommand{\F}[0]{\ensuremath{\mathcal{F}}\xspace}
\renewcommand{\O}[0]{\ensuremath{\mathbb{O}}\xspace}
\newcommand{\Q}[0]{\ensuremath{\mathcal{Q}}\xspace}
\renewcommand{\L}[0]{\ensuremath{\mathcal{L}}\xspace}
\newcommand{\G}[0]{\ensuremath{\mathcal{G}}\xspace}
\newcommand{\X}[0]{\ensuremath{\mathcal{X}}\xspace}
\renewcommand{\P}[0]{\ensuremath{\mathcal{P}}\xspace}

\newcommand{\eqvr}[0]{\ensuremath{\mathcal{R}}\xspace}

\let\plainsquareforqed\squareforqed
\newcommand{\claimqed}{\renewcommand{\squareforqed}{$\diamondsuit$}\qed\renewcommand{\squareforqed}{\plainsquareforqed}}

\newcommand{\pw}[0]{\mathop{\mathrm{\textsc{pw}}}}

\newcommand{\Oh}[0]{\ensuremath{\mathcal{O}}\xspace}

\newlength{\baseImageHeight}
\newcommand{\hyphen}{\nobreakdash-\hspace{0pt}}
\newcommand{\FDeletion}[0]{\textsc{$\F$\hyphen Deletion}\xspace}

\newcommand{\ThreeColoringByComponentSize}[0]{\textsc{$3$\hyphen Coloring [Comp. size]}\xspace}
\newcommand{\PathwidthImprovement}[0]{\textsc{Pathwidth Improvement}\xspace}
\newcommand{\Pathwidth}[0]{\textsc{Pathwidth}\xspace}
\newcommand{\kPathwidth}[0]{\textsc{$k$\hyphen Path\-width}\xspace}
\newcommand{\kRamsey}[0]{\textsc{$k$\hyphen Ramsey}\xspace}
\newcommand{\qColoring}[0]{\textsc{$q$\hyphen Coloring}\xspace}

\newcommand{\treePictures}{
\begin{figure}[t]
\centering
\tikzset{thick,>=stealth,
level distance=4mm,
level 1/.style={sibling distance=10mm},
level 2/.style={sibling distance=3mm},
every node/.style={fill=white,draw=black,circle,inner sep=0,minimum width=0.15cm}}
\tikzset{small/.style={
minimum width=0.01cm,
}}
\subfigure[The graph~$\ttree^2$.]{
\begin{tikzpicture}[thick]
\node[fill=white,draw=white] (dummy) at (0cm, -1.3cm) {};

\node {}  
child foreach \x in {1,2,3} 
	{node {} child foreach \y in {1,2,3} 
			{node{}  }};
			
\end{tikzpicture}
}
\subfigure[The graph~$\tobs^2$.]
{
\begin{tikzpicture}[thick]
\node[fill=white,draw=white] (dummy) at (0cm, -1.3cm) {};

\node {}  
child foreach \x in {1,2,3} 
	{node {} child foreach \y in {1,2,3} 
			{node{} child foreach \w in {1} {node {}} }};

\end{tikzpicture}
}
\subfigure[The graph~$\tobs^2 \diamond 3$.]
{
\begin{tikzpicture}[thick,level 1/.style={sibling distance=14mm},level 2/.style={sibling distance=4mm}]
\node[fill=white,draw=white] (dummy) at (0cm, -1.3cm) {};

\tikzset{every node/.style={fill=white,draw=black,circle,minimum width=0.3cm}};

\node (root) {}
child foreach \x in {1,2,3} 
	{node {} child foreach \y in {1,2,3} 
			{node{} child foreach \w in {1} {node {}} }};

\tikzset{innerVertex/.style={thin,draw=black,fill=black,inner sep=0,minimum width=0.05cm}};

\foreach \i in {root,root-1,root-2,root-3,root-1-1,root-1-2,root-1-3,root-1-1-1,root-1-2-1,root-1-3-1,root-2-1,root-2-2,root-2-3,root-2-1-1,root-2-2-1,root-2-3-1,root-3-1,root-3-2,root-3-3,root-3-1-1,root-3-2-1,root-3-3-1}
{
	\node[innerVertex] (a) at ($(\i) + (90:0.08cm)$) {};
	\node[innerVertex] (b) at ($(\i) + (210:0.08cm)$) {};
	\node[innerVertex] (c) at ($(\i) + (-30:0.08cm)$) {};
	\draw[thin] (a) -- (b);
	\draw[thin] (b) -- (c);
	\draw[thin] (a) -- (c);
}
\end{tikzpicture}
}
\caption{A ternary tree, the corresponding obstruction, and its inflation. All possible edges between connected groups of vertices are present.} \label{img:ternarytrees}
\end{figure}
}

\newcommand{\crossCompositionPicture}{
\begin{figure}[t]
\centering
\begin{tikzpicture}[thick,level distance=6mm,level 1/.style={sibling distance=30mm},level 2/.style={sibling distance=10mm}]
\node[fill=white,draw=white] (dummy) at (0cm, -1.3cm) {};

\tikzset{every node/.style={fill=white,draw=black,circle,minimum width=0.3cm}};

\node (root) {}
child foreach \x in {1,2,3} 
	{node {} child foreach \y in {1,2,3} 
			{node{} child [level distance=8mm] {node {}} }};

\tikzset{innerVertex/.style={thin,draw=black,fill=black,inner sep=0,minimum width=0.05cm}};

\foreach \i in {root,root-1,root-2,root-3,root-1-1,root-1-2,root-1-3,root-1-1-1,root-1-2-1,root-1-3-1,root-2-1,root-2-2,root-2-3,root-2-1-1,root-2-2-1,root-2-3-1,root-3-1,root-3-2,root-3-3,root-3-1-1,root-3-2-1,root-3-3-1}
{
	\node[innerVertex] (a) at ($(\i) + (90:0.08cm)$) {};
	\node[innerVertex] (b) at ($(\i) + (210:0.08cm)$) {};
	\node[innerVertex] (c) at ($(\i) + (-30:0.08cm)$) {};
	\draw[thin] (a) -- (b);
	\draw[thin] (b) -- (c);
	\draw[thin] (a) -- (c);
}

\tikzset{instanceVertex/.style={thin,draw=black,fill=black,inner sep=0,minimum width=0.07cm}};

\node[minimum width=0.8cm] (G1) at (root-1-1-1) {};

\node[instanceVertex] (G1-1) at ($(root-1-1-1) + (0:0.2cm)$) {};
\node[instanceVertex] (G1-2) at ($(root-1-1-1) + (60:0.2cm)$) {};
\node[instanceVertex] (G1-3) at ($(root-1-1-1) + (120:0.2cm)$) {};
\node[instanceVertex] (G1-4) at ($(root-1-1-1) + (180:0.2cm)$) {};
\node[instanceVertex] (G1-5) at ($(root-1-1-1) + (240:0.2cm)$) {};
\node[instanceVertex] (G1-6) at ($(root-1-1-1) + (300:0.2cm)$) {};

\foreach \i \j in {1/2,2/3,3/4,4/5,5/6,6/1}
	\draw[thin] (G1-\i) -- (G1-\j);

\draw[thin] (G1-1) -- (G1-3);
\draw[thin] (G1-5) -- (G1-3);

\node[minimum width=0.8cm] (G2) at (root-1-2-1) {};
\foreach \x \i \j in {1/180/0.1cm,2/0/0.1cm,a/150/0.3cm,b/-150/0.3cm,c/30/0.3cm,d/-30/0.3cm}
	\node[instanceVertex] (G2-\x) at ($(root-1-2-1) + (\i:\j)$) {};
	
\foreach \i \j in {1/2,1/a,1/b,2/c,2/d}
	\draw[thin] (G2-\i) -- (G2-\j);

\node[minimum width=0.8cm] (G3) at (root-1-3-1) {};

\foreach \x \i \j in {1/0/0cm,a/30/0.3cm,b/-30/0.3cm,c/150/0.3cm,d/-150/0.3cm,e/0/0.3cm}
	\node[instanceVertex] (G3-\x) at ($(root-1-3-1) + (\i:\j)$) {};
	
\foreach \i \j in {1/a,1/b,1/c,1/d,c/d,1/e}
	\draw[thin] (G3-\i) -- (G3-\j);

\node[minimum width=0.8cm] (G-2-1) at (root-2-1-1) {};

\node[instanceVertex] (G2-1-1) at ($(root-2-1-1) + (0:0.2cm)$) {};
\node[instanceVertex] (G2-1-2) at ($(root-2-1-1) + (90:0.2cm)$) {};
\node[instanceVertex] (G2-1-3) at ($(root-2-1-1) + (180:0.2cm)$) {};
\node[instanceVertex] (G2-1-4) at ($(root-2-1-1) + (270:0.2cm)$) {};

\foreach \i \j in {1/2,3/2,1/4,3/4,1/3}
	\draw[thin] (G2-1-\i) -- (G2-1-\j);

\node[minimum width=0.8cm] (G-2-2) at (root-2-2-1) {};

\node[instanceVertex] (G2-2-1) at ($(root-2-2-1) + (0:0.2cm)$) {};
\node[instanceVertex] (G2-2-2) at ($(root-2-2-1) + (60:0.2cm)$) {};
\node[instanceVertex] (G2-2-3) at ($(root-2-2-1) + (120:0.2cm)$) {};
\node[instanceVertex] (G2-2-4) at ($(root-2-2-1) + (180:0.2cm)$) {};
\node[instanceVertex] (G2-2-5) at ($(root-2-2-1) + (240:0.2cm)$) {};
\node[instanceVertex] (G2-2-6) at ($(root-2-2-1) + (300:0.2cm)$) {};

\foreach \i \j in {1/2,2/3,3/4,4/5,5/6,6/1,2/4,2/5}
	\draw[thin] (G2-2-\i) -- (G2-2-\j);

\node[minimum width=0.8cm] (G-2-3) at (root-2-3-1) {};

\node[instanceVertex] (G2-3-1) at ($(root-2-3-1) + (0:0cm)$) {};
\node[instanceVertex] (G2-3-4) at ($(root-2-3-1) + (270:0.15cm)$) {};
\node[instanceVertex] (G2-3-3) at ($(root-2-3-1) + (270:0.3cm)$) {};

\node[instanceVertex] (G2-3-2) at ($(root-2-3-1) + (120:0.2cm)$) {};
\node[instanceVertex] (G2-3-5) at ($(root-2-3-1) + (60:0.2cm)$) {};

\foreach \i \j in {2/5,1/2,1/5,1/4,4/3}
	\draw[thin] (G2-3-\i) -- (G2-3-\j);

\node[minimum width=0.8cm] (G-3-1) at (root-3-1-1) {};

\node[instanceVertex] (G3-1-1) at ($(root-3-1-1) + (0:0.2cm)$) {};
\node[instanceVertex] (G3-1-2) at ($(root-3-1-1) + (60:0.2cm)$) {};
\node[instanceVertex] (G3-1-3) at ($(root-3-1-1) + (120:0.2cm)$) {};
\node[instanceVertex] (G3-1-5) at ($(root-3-1-1) + (240:0.2cm)$) {};
\node[instanceVertex] (G3-1-6) at ($(root-3-1-1) + (300:0.2cm)$) {};

\foreach \i \j in {1/2,2/3,5/6,6/1}
	\draw[thin] (G3-1-\i) -- (G3-1-\j);

\draw[thin] (G3-1-1) -- (G3-1-3);
\draw[thin] (G3-1-5) -- (G3-1-3);

\node[minimum width=0.8cm] (G-3-2) at (root-3-2-1) {};

\node[instanceVertex] (G3-2-1) at ($(root-3-2-1) + (0:0.2cm)$) {};
\node[instanceVertex] (G3-2-3) at ($(root-3-2-1) + (180:0.2cm)$) {};
\node[instanceVertex] (G3-2-4) at ($(root-3-2-1) + (270:0.2cm)$) {};

\node[instanceVertex] (G3-2-2) at ($(root-3-2-1) + (120:0.2cm)$) {};
\node[instanceVertex] (G3-2-5) at ($(root-3-2-1) + (60:0.2cm)$) {};

\foreach \i \j in {1/4,3/4,1/3,2/5}
	\draw[thin] (G3-2-\i) -- (G3-2-\j);

\node[minimum width=0.8cm] (G-3-3) at (root-3-3-1) {};

\foreach \x \i \j in {1/0/0cm,a/30/0.3cm,b/-30/0.3cm,c/150/0.3cm,d/-150/0.3cm}
	\node[instanceVertex] (G3-\x) at ($(root-3-3-1) + (\i:\j)$) {};
	
\foreach \i \j in {1/a,1/b,1/c,1/d,c/d,a/b}
	\draw[thin] (G3-\i) -- (G3-\j);

\end{tikzpicture}
\caption{Result of \orsymb-cross-composing nine inputs with~$k=3$ into one.} \label{img:crosscomposition}
\end{figure}
}

\title{FPT is Characterized by\\Useful Obstruction Sets\thanks{This work was supported by the Netherlands Organization for Scientific Research (NWO), project ``KERNELS: Combinatorial Analysis of Data Reduction''.}}

\author{Michael R.\ Fellows \inst{1} \and Bart M.\ P.\ Jansen \inst{2}}

\institute{
Charles Darwin University, Australia. \email{Michael.Fellows@cdu.edu.au}
\and
Utrecht University, The Netherlands.
\email{B.M.P.Jansen@uu.nl}
}

\begin{document}

\hypersetup{bookmarksdepth=-1}

\maketitle

\hypersetup{bookmarksdepth=2}

\begin{abstract}
Many graph problems were first shown to be fixed-parameter tractable using the results of Robertson and Seymour on graph minors. We show that the combination of finite, computable, obstruction sets and efficient order tests is not just one way of obtaining strongly uniform FPT algorithms, but that \emph{all} of FPT may be captured in this way. Our new characterization of FPT has a strong connection to the theory of kernelization, as we prove that problems with polynomial kernels can be characterized by obstruction sets whose elements have polynomial size. Consequently we investigate the interplay between the sizes of problem kernels and the sizes of the elements of such obstruction sets, obtaining several examples of how results in one area yield new insights in the other. We show how exponential-size minor-minimal obstructions for pathwidth~$k$ form the crucial ingredient in a novel \orsymb-cross-composition for \kPathwidth, complementing the trivial \andsymb-composition that is known for this problem. In the other direction, we show that \orsymb-cross-compositions into a parameterized problem can be used to rule out the existence of efficiently generated quasi-orders on its instances that characterize the \no-instances by polynomial-size obstructions.
\end{abstract}

\begin{bibunit}[abbrv]

\section{Introduction}
This paper is concerned with the connection between fixed-parameter tractability, kernelization, and the characterization of parameterized problems by efficiently testable obstruction sets. Historically, this connection has been a major impetus to the development of the field of parameterized complexity. The results of the Graph Minors project were applied to obtain some of the first classifications~\cite{FellowsL88a} of problems as (nonuniformly) fixed-parameter tractable. Robertson and Seymour proved that the set of unlabeled finite graphs is well-quasi-ordered by the minor relation~\cite{RobertsonS04}. By standard well-quasi-order theory, this implies that any set of graphs~$\F$ that is closed under taking minors (a \emph{lower ideal in the minor order}) is characterized by a \emph{finite} obstruction set~$\O_\F$ in the following sense: a graph is contained in~$\F$ if and only if it does not contain an element of~$\O_\F$ as a minor. They also provided an algorithm for each fixed graph~$H$ that tests, given a graph~$G$, whether~$H$ is a minor of~$G$ in~$\Oh(n^3)$ time~\cite{RobertsonS95b}. 

The algorithmic implications of this machinery are well known. Consider a parameterized graph problem~$\Q$ whose input consists of a graph~$G$ and integer~$k$. If~$\Q$ is minor-closed, that is, if~$(G',k)$ is a \yes-instance whenever~$(G,k)$ is a \yes-instance and~$G'$ is a minor of~$G$, then~$\Q$ can be solved in~$\Oh(n^3)$ time for each fixed~$k$. As the \yes-instances of a fixed parameter value~$k$ form a minor ideal, there is a finite obstruction set~$\O_k$ that characterizes the ideal. Thus we can decide whether~$(G,k) \in \Q$ by testing for each graph in~$\O_k$ whether it is a minor of~$G$. By deriving an algorithm to compute the obstruction sets~$\O_k$, this approach yields constructive, uniform FPT algorithms (cf.~\cite[\S 7.9.2]{DowneyF99}).

Our first result in this paper shows that the described tools for developing FPT algorithms --- efficient order tests for quasi-orders that characterize the \yes-instances of a fixed parameter value by finite obstructions sets --- are not just \emph{one} way of obtaining (strongly uniform) FPT characterizations, but that in fact all of FPT can be characterized in this way. For this general result we relax from the minor order and instead consider arbitrary quasi-orders on the set of instances~$\Sigma^* \times \mathbb{N}$ of a parameterized problem (see Section~\ref{section:preliminaries} for definitions).

We introduce some terminology to state the characterization. A quasi-order is a reflexive and transitive binary relation~$\preceq$ on a set~$S$. For elements~$x,y \in S$ such that~$x \preceq y$ we say that \emph{$x$ precedes~$y$}. If~$x$ precedes~$y$ and~$x \neq y$ then~\emph{$x$ strictly precedes~$y$}, denoted~$x \prec y$. A quasi-order~$\preceq$ is \emph{polynomial-time} if there is an algorithm that decides whether~$x \preceq y$ in~$\Oh((|x|+|y|)^{\Oh(1)})$ time. If~$S$ is a subset of a universe~$U$ and~$\preceq$ is a quasi-order on~$U$, then~$S$ is a \emph{lower ideal} of~$U$ if~$x \in S$ and~$x' \preceq x$ together imply that~$x' \in S$. Our characterization extends the folklore result stating that all problems in FPT have kernels. 

\begin{theorem} \label{theorem:fptCharacterization}
For any parameterized problem~$\Q \subseteq \Sigma^* \times \mathbb{N}$, the following statements are equivalent:
\begin{enumerate}
	\item Problem~$\Q$ is strongly uniformly fixed-parameter tractable.\label{characterization:fpt}
	\item Problem~$\Q$ is decidable and admits a kernel whose size is computable.\label{characterization:kernel}
	\item Problem~$\Q$ is decidable and there is a polynomial-time quasi-order~$\preceq$ on~$\Sigma^* \times \mathbb{N}$ and a computable function~$f \colon \mathbb{N} \to \mathbb{N}$ such that: \label{characterization:quasiorder}
	\begin{enumerate}[label=\alph*.,ref=(\arabic{enumi}.\alph{enumii})]
		\item The set~$\Q$ is a lower ideal of~$\Sigma^* \times \mathbb{N}$ under~$\preceq$.\label{characterization:lowerideal}
		\item For every~$(x,k) \not \in \Q$ there is an \emph{obstruction}~$(x',k') \not \in \Q$ of size at most~$f(k)$ with~$(x',k') \preceq (x,k)$.\label{characterization:obstruction}
	\end{enumerate}
\end{enumerate}
\end{theorem}

Let us make some remarks about the theorem. Criterion~\ref{characterization:obstruction} is stated in terms of small obstructions rather than finite, computable obstruction sets, to make the subsequent theorem that proves the \emph{non-existence} of such quasi-orders (Theorem \ref{theorem:crosscomposition:obstruction:bounds}) stronger. The existence of computable obstruction sets follows directly from the given conditions, as will indeed be exploited in the proof of Theorem~\ref{theorem:fptCharacterization} in Section~\ref{section:characterization}. The proof also shows that problems with kernels of size~$f$ are characterized by obstructions of size~$f$ under polynomial-time quasi-orders. Hence problems with polynomial kernels can be characterized by polynomial-size obstructions. This general quantitative connection between kernel sizes and obstruction sizes leads us to investigate the relationship between the two in more concrete settings. While a construction due to Kratsch and Wahlstr\"om~\cite{KratschW11z} shows that it is unlikely that all problems characterized by polynomial-size obstructions have polynomial kernels, there is a rich interaction between the two domains, which forms the topic of the remainder of this work.

\textbf{A Cross-Composition Based on Large Obstructions.}
In Section~\ref{section:or:composition:for:pw} we give an example of how properties of obstruction sets can be exploited to obtain kernel bounds. Our example concerns the \kPathwidth problem, which asks whether the pathwidth of a given graph~$G$ is at most~$k$. For any sequence of graphs~$G_1, \ldots, G_t$, the disjoint union~$G_1 \dotcup G_2 \dot \cup \ldots \dot \cup G_t$ has pathwidth at most~$k$, if and only if each~$G_i$ has pathwidth at most~$k$. Hence there is a trivial \andsymb-composition~\cite{BodlaenderDFH09} for \kPathwidth. Using existing methods~\cite{BodlaenderDFH09,Drucker12} this proves that \kPathwidth does not admit a polynomial kernel unless \containment.

The majority of kernelization lower bounds currently known, however, are not obtained by \andsymb-composition but by \orsymb-(cross-)-composition~\cite{BodlaenderDFH09,BodlaenderJK11}: polynomial-time algorithms that take a sequence of instances as input, and output a single instance of bounded parameter value whose answer is \yes if and only if \emph{at least one} (rather than all) of the inputs are \yes-instances. Given the nature of the pathwidth problem, it seems to lend itself much better to \andsymb-composition than to \orsymb-(cross-)-composition. However, we show that an \orsymb-cross-composition into \kPathwidth can be obtained by embedding instances of a related problem into a minor-obstruction for pathwidth~$k$ containing~$\Theta(3^k)$ vertices. The properties of obstructions are exploited to ensure the correctness of this construction. The fact that the size of the obstruction is exponential in~$k$, is crucial to obtaining this superpolynomial kernelization lower bound. The construction illustrates how properties of obstruction sets can be used to obtain kernelization bounds.

\textbf{Bounds on Obstruction Sizes by Cross-Composition.}
We study how kernel bounds may be used to derive properties of obstruction sets in Section~\ref{section:crosscomposition:obstruction:bounds}. The \orsymb-cross-composition framework for kernelization lower bounds turns out to have interesting connections to obstruction sizes. We introduce the notion of an \emph{efficiently generated} quasi-order, which, roughly speaking, is a quasi-order such that the elements preceding a given instance~$(x,k)$ can appear on the output paths of a polynomial-time nondeterministic Turing machine. If there is an efficiently generated quasi-order on the instances of a parameterized problem, such that each \no-instance~$(x,k)$ is preceded by a \no-instance of size~$f(k)$ (an obstruction), then this results in a nondeterministic form of kernel, of size~$f(k)$. As an \orsymb-cross-composition together with a polynomial kernel implies that~\containment~\cite{BodlaenderJK11}, even in the nondeterministic setting~\cite{KratschPRR12}, this gives us the means to prove that certain parameterized problems are unlikely to be characterized by efficiently generated quasi-orders with polynomial-size obstructions. Using our \orsymb-cross-composition for \kPathwidth we can conclude that obstructions to \kPathwidth are not only of superpolynomial size in the minor order, but must be of superpolynomial size for all efficiently generated quasi-orders under which \kPathwidth is closed. Other examples of the connection between kernels and obstructions are discussed in Section~\ref{section:conclusion}.

\textbf{Related Work.} There are many alternative characterizations of FPT, as described for example by Flum and Grohe~\cite[\S 1.6]{FlumG06}. Obstruction sets form a popular topic of study (e.g.,~\cite{DinneenCF01,DinneenL07,Kinnersley92,RueST12,TakahashiUK94}). The task of computing obstruction sets has also been investigated thoroughly (e.g.,~\cite{CattellDDFL00,FellowsL89,Lagergren98}). Dinneen~\cite[Theorem 5]{Dinneen97} related properties of obstruction sets to complexity-theoretic assumptions. He showed that the number of elements in obstruction sets corresponding to NP-hard minor-closed graph problems with parameter~$k$ cannot be polynomial in~$k$, unless \containment.

\section{Preliminaries} \label{section:preliminaries}
\textbf{Parameterized complexity and kernels.}
A parameterized problem~$\Q$ is a subset of~$\Sigma^* \times \mathbb{N}$, the second component being the \emph{parameter}. For an instance~$(x,k) \in \Sigma^* \times \mathbb{N}$ we define the \emph{size} of~$(x,k)$ to be~$|(x,k)| := |x|+k$. A parameterized problem is (strongly uniformly) \emph{fixed-parameter tractable} if there exists an algorithm to decide whether $(x,k) \in \Q$ in time~$f(k)|x|^{\Oh(1)}$ where~$f$ is a computable function. We refer to the textbooks~\cite{DowneyF99,FlumG06} for more background on parameterized complexity.

A \emph{kernelization algorithm} (or \emph{kernel}) of size~$f \colon \mathbb{N} \to \mathbb{N}$ for a parameterized problem~$\Q \subseteq \Sigma^* \times \mathbb{N}$ is a polynomial-time algorithm that, on input~$(x,k) \in \Sigma^* \times \mathbb{N}$, outputs an instance~$(x', k')$ of size at most~$f(k)$ such that~$(x,k) \in \Q \Leftrightarrow (x', k') \in \Q$. If~$f(k) \in \Oh(k^{\Oh(1)})$ then this is a \emph{polynomial kernel} (cf.~\cite{Bodlaender09}). 

\textbf{Cross-composition.} We use the framework of cross-composition to prove kernel lower bounds, including the definition of a \emph{polynomial equivalence relation} and a \emph{cross-composition} as given by Bodlaender et al.~\cite{BodlaenderJK11}. These definitions are repeated in Appendix~\ref{appendix:crosscomposition} for completeness. To highlight the differences between \orsymb and \andsymb compositions, we call the type cross-composition defined by Bodlaender et al.~\cite{BodlaenderJK11} \emph{\orsymb-cross-composition}.

\begin{theorem}[{\cite{BodlaenderJK11}}] \label{theorem:crossCompositionNoKernel}
If a set~$\L \subseteq \Sigma^*$ is NP-hard under Karp reductions and~$\L$ \orsymb-cross-composes into the parameterized problem~$\Q$, then there is no polynomial kernel for~$\Q$ unless \containment.
\end{theorem}

Theorem~\ref{theorem:crossCompositionNoKernel} has been extended to the co-nondeterministic setting in recent publications. Kratsch et al.~\cite[Theorem 2]{KratschPRR12} exploited the fact that the lower bound machinery also works if the cross-composition is co-nondeterministic.

\textbf{Graphs.}
All graphs we consider are finite, simple, and undirected. An undirected graph~$G$ consists of a vertex set~$V(G)$ and an edge set~$E(G)$, whose members are 2-element subsets of~$V(G)$. We write~$G \subseteq H$ if graph~$G$ is a subgraph of graph~$H$. 
The \emph{clique number}~$\omega(G)$ of~$G$ is the size of a largest clique in~$G$. For a set of vertices~$X$ in a graph~$G$ we use~$G - X$ to denote the graph that results after deleting all vertices of~$X$ and their incident edges. When deleting a single vertex~$v$, we write~$G-v$ rather than~$G - \{v\}$. 
Graph~$H$ is a \emph{minor} of graph~$G$ if~$H$ can be obtained from a subgraph of~$H$ by edge contractions. If~$H \neq G$ is a minor of~$G$, then~$H$ is a \emph{proper minor} of~$G$. A vertex of degree at most one is a \emph{leaf}.

A \emph{path decomposition} of a graph~$G$ is a sequence $(\X_1, \ldots, \linebreak[0] \X_r)$ of subsets of~$V(G)$, called \emph{bags}, such that: (i) $\bigcup _{i \in [r]} \X_i = V(G)$, (ii) for each edge~$\{u,v\} \in E(G)$ there is a bag~$\X_i$ containing~$v$ and~$w$, and (iii) for each~$v \in V(G)$, the bags containing~$v$ are consecutive in the sequence. The \emph{width} of a path decomposition is~$\max _{1 \leq i \leq r} |\X_i| - 1$. The \emph{pathwidth} of a graph~$G$, denoted~$\pw(G)$, is the minimum width over all path decompositions of~$G$. We say that an edge~$\{u,v\}$ is \emph{realized} by any bag that contains~$u$ and~$v$. Condition~(iii) is also called the \emph{convexity property} of path decompositions. The set~$\{1, 2, \ldots, n\}$ is abbreviated as~$[n]$.

\section{Characterizing Problems in FPT by Small Obstructions}\label{section:characterization}
In this section we present the proof of Theorem~\ref{theorem:fptCharacterization} and consider some of its consequences.

\begin{proof}[of Theorem~\ref{theorem:fptCharacterization}]
Let~$\Q$ be a parameterized problem. It is well-known that conditions~\ref{characterization:fpt} and~\ref{characterization:kernel} are equivalent~\cite[Theorem 1]{Bodlaender09}. We prove that~\ref{characterization:quasiorder}$\Rightarrow$\ref{characterization:fpt} and that~\ref{characterization:kernel}$\Rightarrow$\ref{characterization:quasiorder}.

\ref{characterization:quasiorder}$\Rightarrow$\ref{characterization:fpt}. Consider a combination of~$\preceq$ and~$f \colon \mathbb{N} \to \mathbb{N}$ that satisfies the preconditions to~\ref{characterization:quasiorder}. We obtain an FPT algorithm that decides~$\Q$ by showing that there is an algorithm that computes bounded-size obstruction sets to membership in~$\Q$. Let~$k \in \mathbb{N}$ and define~$O_k$ as the \no-instances of~$\Q$ that have size at most~$f(k)$. Let~$\O_k$ be the elements of~$O_k$ that are minimal under~$\preceq$, i.e., those elements of~$O_k$ that are not preceded by another element of~$O_k$.

\begin{claim}
Let~$k \in \mathbb{N}$. For any~$x \in \Sigma^*$ we have~$(x,k) \in \Q$ if and only if there is no element in~$\O_k$ that precedes~$(x,k)$.
\end{claim}
\begin{proof}
Fix some~$k \in \mathbb{N}$ and consider some~$x \in \Sigma^*$. If~$(x,k)$ is a \yes-instance then all elements that precede it under~$\preceq$ are \yes-instances, by~\ref{characterization:lowerideal}. If~$(x,k)$ is a \no-instance, then by~\ref{characterization:quasiorder} there is an obstruction~$(x',k')$ of size at most~$f(k)$ that is a \no-instance of~$\Q$ and precedes~$(x,k)$. But then there is a minimal \no-instance with these properties, which is contained in~$\O_k$ by definition. Hence there is an element of~$\O_k$ that precedes~$(x,k)$.
\claimqed
\end{proof}

There is an algorithm that, on input~$k \in \mathbb{N}$, computes the set~$\O_k$: this follows from the facts that~$\Q$ is decidable,~$f$ is computable, and~$\preceq$ is polynomial-time. From the algorithm that computes the obstruction sets~$\O_k$ we obtain a strongly uniformly fixed-parameter tractable algorithm for~$\Q$, as follows. On input~$(x,k)$, compute the set~$\O_k$. Test if there is an obstruction in~$\O_k$ that precedes~$(x,k)$ using the order testing algorithm for~$\preceq$. By the claim, the answer to~$(x,k)$ is \yes if and only if there is no such preceding element. The running time is bounded by~$g(k)|x|^{\Oh(1)}$ for some computable function~$g$: the time to compute~$\O_k$ is computable, while the~$|\O_k|$ order tests take~$\Oh((f(k) + |(x,k)|)^{\Oh(1)})$ time each.

\ref{characterization:kernel}$\Rightarrow$\ref{characterization:quasiorder}. Let~$K$ be a kernelization algorithm for~$\Q$ that maps instances~$(x,k)$ to equivalent instances~$(x',k')$ of size at most~$f$, for some computable function~$f$. We define a polynomial-time quasi-order~$\preceq$ by giving an algorithm that decides, given~$(x,k)$ and~$(x',k')$, whether~$(x',k') \preceq (x,k)$. The algorithm proceeds as follows. If~$(x',k') = (x,k)$ then it immediately output \yes. Otherwise, it sets~$(x^*, k^*) := (x,k)$. While~$|K(x^*,k^*)| < |(x^*, k^*)|$ it replaces~$(x^*, k^*)$ by~$K(x^*, k^*)$, i.e., it repeatedly applies the kernelization algorithm until this no longer decreases the total size of the instance. It then outputs \yes if and only if~$(x',k')$ equals the resulting instance~$(x^*,k^*)$.

\begin{claim}
The relation~$\preceq$ defined by the algorithm is a polynomial quasi-order and~$\Q$ is a lower ideal under~$\preceq$.
\end{claim}
\begin{proof}
The number of iterations made by the algorithm on inputs~$(x,k)$ and $(x',k')$ is bounded by~$|x| + k$, as the length of the instance is decreased in each iteration. As each invocation of~$K$ takes polynomial time, the entire comparison algorithm executes in polynomial time. 

It is obvious that~$\preceq$ is reflexive. To prove that it is a quasi-order, it remains to prove transitivity. Consider three instances such that~$(x'',k'') \preceq (x',k') \preceq (x,k)$. We shall prove that~$(x'',k'') = (x',k')$ or~$(x',k') = (x,k)$, which obviously implies that~$(x'',k'') \preceq (x,k)$. So assume that~$(x',k') \neq (x,k)$. By definition of the algorithm that decides~$\preceq$, it then follows that~$(x',k')$ is the unique instance that is obtained from~$(x,k)$ by repeatedly applying the kernelization algorithm~$K$ until it no longer strictly shrinks the size of the instance. Hence for~$(x',k')$ we know that~$|K(x',k')| \geq |(x',k')|$. But then any instance~$(x^*,k^*)$ with~$(x^*,k^*) \preceq (x',k')$ must be identical to~$(x',k')$, by that same definition. Thus~$(x'',k'') = (x',k')$, which implies that~$(x'',k'') \preceq (x,k)$. Hence~$\preceq$ is transitive.

Finally let us establish that~$\Q$ is a lower ideal of~$\Sigma^* \times \mathbb{N}$ under~$\preceq$. Since a kernelization maps an instance to an equivalent instance, it is easily seen that if~$(x',k') \preceq (x,k)$ then~$(x',k') \in \Q \Leftrightarrow (x,k) \in \Q$. Hence~$(x,k) \in \Q$ and~$(x',k') \preceq (x,k)$ together imply that~$(x',k') \in \Q$.
\claimqed
\end{proof}

\begin{claim}
For every~$(x,k) \not \in \Q$ there is an \emph{obstruction}~$(x',k') \not \in \Q$ of size at most~$f(k)$ with~$(x',k') \preceq (x,k)$.
\end{claim}
\begin{proof}
Consider some~$(x,k) \not \in \Q$. Let~$(x',k')$ be the result of applying kernelization~$K$ to the instance, as long as its total size decreases by this operation. By definition of~$\preceq$ we have~$(x',k') \preceq (x,k)$. Since the kernelization preserves the membership status in~$\Q$ we find that~$(x',k')$ is a \no-instance. Since~$K$ is a kernel of size~$f(k)$ we have~$|K(x,k)| \leq f(k)$, which implies that~$|(x',k')| \leq f(k)$.
\claimqed
\end{proof}

The two claims show that the combination of~$\preceq$ and the function~$f$ satisfy the requirements of property~\ref{characterization:quasiorder}, concluding the proof.
\qed
\end{proof}

In the proof of Theorem~\ref{theorem:fptCharacterization}, the size of the obstructions of~\ref{characterization:obstruction} matches the size bound of the kernel from which the quasi-order~$\preceq$ is derived. Hence problems with polynomial kernels can be characterized by polynomial-size obstructions.

\begin{corollary} \label{corollary:kernel:yields:small:obstructions}
If~$\Q \subseteq \Sigma^* \times \mathbb{N}$ is a decidable parameterized problem with a kernel of size~$\Oh(k^c)$, then there is a polynomial-time quasi-order~$\preceq$ on~$\Sigma^* \times \mathbb{N}$ and a function~$f$ that together satisfy statement~\ref{characterization:quasiorder} of Theorem~\ref{theorem:fptCharacterization}, with~$f(k) \in \Oh(k^c)$.
\end{corollary}

It follows from a construction by Kratsch and Wahlstr\"om~\cite{KratschW11z} that the converse of Corollary~\ref{corollary:kernel:yields:small:obstructions} is false, assuming \ncontainment. We give a concrete example of a problem that is characterized by efficiently testable obstructions of polynomial size, yet is unlikely to admit a polynomial kernel.

\parproblemdef{\ThreeColoringByComponentSize}
{An undirected graph~$G$ and an integer~$k$ that bounds the maximum size of a connected component in~$G$.}
{$k$.}
{Is there a proper $3$-coloring of the vertices of~$G$?}

\begin{lemma}[$\bigstar$, Cf.~\cite{KratschW11z}] \label{lemma:obstructions:donot:imply:kernel}
\ThreeColoringByComponentSize does not admit a polynomial kernel unless \containment, but there is a polynomial-time quasi-order on its instances that satisfies statement~\ref{characterization:quasiorder} of Theorem~\ref{theorem:fptCharacterization} with~$f(k) \in \Oh(k^2)$.
\end{lemma}

\section{OR-Cross-Composition into \texorpdfstring{$k$}{k}-Pathwidth} \label{section:or:composition:for:pw}
A \emph{minor-minimal obstruction to pathwidth~$k$} is a graph of pathwidth~$k+1$, such that all its proper minors have pathwidth~$\leq k$. Minor-minimal obstructions to pathwidth~$k$ of size~$\Theta(3^k)$ form the crucial ingredient for an \orsymb-cross-composition of an NP-complete problem into \kPathwidth. The following \emph{improvement} version of the problem serves as the starting point for the composition.

\problemdef{\PathwidthImprovement}
{A graph~$G$, an integer~$k$ with~$2 \leq k \leq |V(G)|$, and a path decomposition~\P of~$G$ having width~$k-1$.}
{Is the pathwidth of~$G$ at most~$k-2$?}

\begin{lemma}[$\bigstar$] \label{lemma:improvement:npcomplete}
\PathwidthImprovement is NP-complete.
\end{lemma}

\noindent The path decomposition in the input of \PathwidthImprovement makes it possible to verify in polynomial time that the pathwidth of the graph does not exceed~$k-1$. The additive terms are chosen to simplify the correctness proof of the \orsymb-cross-composition. The exponential-size obstructions to pathwidth that we need for our construction are defined as follows.

\begin{definition}
For~$i \in \mathbb{N}_0$, let~$\ttree^i$ denote the complete ternary tree of height~$i$ with~$3^i$ leaves. Let~$\tobs^i$ be the graph obtained from~$\ttree^i$ by adding, for each leaf~$v$ of~$\ttree^i$, a new vertex that is only adjacent to~$v$.
\end{definition}

\begin{lemma}[$\bigstar$, Cf.~\cite{Kinnersley92,TakahashiUK94}] \label{lemma:ternarytrees:obstructions}
For~$k \in \mathbb{N}_0$ the graph~$\tobs^k$ is a minor-minimal obstruction to pathwidth~$k$.
\end{lemma}

\noindent In our \orsymb-cross-composition, we need to \emph{inflate} obstructions before being able to embed a series of input instances into them.

\begin{definition}
Let~$G$ be a graph and let~$k \in \mathbb{N}$. The graph~$G \diamond k$, called the \emph{inflation of~$G$ by~$k$}, is defined as follows:
\begin{itemize}
	\item $V(G \diamond k) := \bigcup _{v \in V(G)} \{ v_1, \ldots, v_k \}$.
	\item Vertices~$u_i$ and~$v_j$ are adjacent in~$G \diamond k$ if~$u = v$ or~$\{u, v\} \in E(G)$.
\end{itemize}
For a vertex~$v \in V(G)$ we call the vertices~$v_1, \ldots, v_k$ in~$G \diamond k$ the \emph{copies of~$v$}.
\end{definition}

\treePictures

\noindent Refer to Fig.~\ref{img:ternarytrees} for an example. Inflation of a graph has a straight-forward effect on its pathwidth.

\begin{lemma}[$\bigstar$] \label{lemma:pathwidth:inflation}
For any graph~$G$ and~$k \in \mathbb{N}: \pw(G \diamond k) + 1 = k \cdot (\pw(G) + 1)$.
\end{lemma}

\begin{theorem}[$\bigstar$] \label{theorem:pathwidth:crosscomposition}
The \PathwidthImprovement problem \orsymb-cross-composes into \kPathwidth.
\end{theorem}
\begin{proof}[Sketch]
Using a suitable choice of polynomial equivalence relation, permitted by the cross-composition framework~\cite{BodlaenderJK11}, it suffices to give a polynomial-time algorithm of the following form. The input is a sequence~$(G_1, k, \P^1), \linebreak[0] \ldots, \linebreak[0] (G_t, k, \P^t)$ of instances of \PathwidthImprovement that all share the same value of~$k$, and the output is a single instance~$(G', k')$ of \kPathwidth, with~$k'$ polynomial in~$\max _{i \in [t]} |V(G_i)| + \log t$, such that~$\pw(G') \leq k'$ if and only if there is a \yes-instance among the inputs. By standard arguments we may assume that~$t$ is a power of three, so let~$t = 3^s$ for~$s \in \mathbb{N}$.

The construction of~$(G',k')$ is based on the minor-minimal obstruction~$\tobs^s$. Label the~$3^s = t$ leaves of~$\tobs^s$ as~$x^1, \ldots, x^t$, and let~$y^1, \ldots, y^t$ be the parents of those leaves. As~$s \geq 1$ each vertex~$y_i$ has degree exactly two in~$\tobs^s$. We cross-compose the instances into a single graph~$G'$. It is obtained by inflating~$\tobs^s$ by a factor~$k$ and replacing each $k$-vertex clique containing the copies of a leaf~$x_i$ by the graph~$G_i$. More formally, we obtain~$G'$ as follows.
\begin{itemize}
	\item Initialize~$G'$ as the inflation~$\tobs^s \diamond k$. For each leaf~$x^i$ the copies created by the inflation form a clique of size~$k$ on vertices~$x^i_1, \ldots, x^i_k$.
	\item For each~$i \in [t]$, remove the vertices~$x^i_1, \ldots, x^i_k$ from~$G'$ and replace them by a copy of the graph~$G_i$. Make all vertices of~$G_i$ adjacent to the copies of the parent of~$x^i$, i.e., to the vertices~$y^i_1, \ldots, y^i_k$.
\end{itemize}
Refer to Fig.~\ref{img:crosscomposition} for an example. Let~$k' := k (s+2) - 2 \in \Oh(\max _{i \in [t]} |V(G)| \cdot \log t)$.

\crossCompositionPicture

\begin{claim}
$\pw(G') \leq k'$ if and only if there is an~$i \in [t]$ such that~$\pw(G_i) \leq k - 2$.
\end{claim}
The ``if'' direction is proven as follows. Suppose that for the instance~$(G_i, k, \P^i)$ of \PathwidthImprovement we have~$\pw(G_i) \leq k-2$. As~$\pw(\tobs^s - x_i) = s$ by Lemma~\ref{lemma:ternarytrees:obstructions}, the inflation satisfies~$\pw((\tobs^s - x_i) \diamond k) + 1 = k \cdot (s + 1)$ by Lemma~\ref{lemma:pathwidth:inflation}. From a path decomposition~$\P'$  of~$(\tobs^s - x_i) \diamond k$ we obtain a path decomposition of~$G' - V(G_i)$ of the same width: for each inserted instance~$G_j$ with~$j \neq i$ the vertices~$x^j_1, \ldots, x^j_k, y^j_1, \ldots, y^j_k$ form a clique in~$\tobs^s - x_i$. Hence~$\P'$ has a bag containing all those vertices, and we may replace~$x^j_1, \ldots, x^j_k$ by the width-$(k-1)$ decomposition~$\P^j$ of~$G_j$ without increasing the width. From a path decomposition of~$G' - V(G_i)$ we obtain a decomposition of~$G'$ by inserting a width-$(k-2)$ decomposition for~$G_i$ in the appropriate place, increasing the total width by~$k-1$ to~$k \cdot (s + 1) - 1 + (k-1) = k \cdot (s+2) - 2 = k'$.

The ``only if' direction of the claim is proven by contraposition. We use that the pathwidth of a graph equals~$\min _H(\omega(H) - 1)$ over its interval supergraphs~$H$. Suppose all inputs have pathwidth at least~$k-1$, implying all interval supergraphs of the inputs have clique number at least~$k$. An interval supergraph~$H'$ of~$G'$ contains interval supergraphs of~$G_1, \ldots, G_t$. As the latter all contain a clique of size at least~$k$, graph~$H'$ is a supergraph of~$\tobs^s \diamond k$. But then~$H'$ has pathwidth at least~$k \cdot (s+2) - 1$ by Lemmata~\ref{lemma:ternarytrees:obstructions} and~\ref{lemma:pathwidth:inflation}, implying the same for~$G'$.

The claim shows that~$(G',k')$ acts as the \orsymb of the inputs. As it can be built in polynomial time, this concludes the proof.
\qed
\end{proof}

Lemma~\ref{lemma:improvement:npcomplete} and Theorem~\ref{theorem:pathwidth:crosscomposition} provide a new way of proving that \kPathwidth does not admit a polynomial kernel unless \containment, by Theorem~\ref{theorem:crossCompositionNoKernel}.

\section{Proving Nonexistence of Small Obstructions} \label{section:crosscomposition:obstruction:bounds}

In this section we show how the kernelization lower-bound framework of \orsymb-cross-composition can be used to prove that a problem is \emph{not} characterized by polynomial-size obstructions under any quasi-order of the following form.

\begin{definition} \label{definition:efficiently:generated}
A quasi-order~$\preceq$ on~$\Sigma^* \times \mathbb{N}$ is \emph{efficiently generated} if there is a polynomial-time nondeterministic Turing machine that, on input~$(x,k) \in \Sigma^* \times \mathbb{N}$, outputs an instance~$(x',k') \in \Sigma^* \times \mathbb{N}$ on each computation path such that:
\begin{itemize}[topsep=2pt, partopsep=2pt]
	\item each output instance~$(x',k')$ precedes~$(x,k)$ under~$\preceq$, and
	\item all instances preceding~$(x,k)$ appear as the output of some computation path.
\end{itemize}
\end{definition}

\noindent Many well-known containment relations on graphs are efficiently generated. As a concrete example, consider the relation on parameterized graphs~$(G,k)$ encoded by adjacency matrices, where~$(G',k) \preceq (G,k)$ if~$G'$ is a minor of~$G$. This order is efficiently generated. A NDTM nondeterministically selects a subgraph~$G'$ of its input~$(G,k)$, then selects a set of edges to contract to obtain the minor~$G''$, and outputs~$(G'',k)$. By using a nondeterministically selected order on the vertices when encoding~$G''$ as an adjacency matrix, all isomorphism classes of the minor~$G''$ are generated. The correctness of the procedure is easy to verify.

Other efficiently generated quasi-orders on parameterized graphs, encoded as adjacency matrices, include the topological minor order, the (induced) subgraph order, the immersion order, and the contraction order (cf.~\cite[\S 7.8]{DowneyF99}). The quasi-order constructed in the proof of Theorem~\ref{theorem:fptCharacterization} is also efficiently generated. %

The following lemma, along with the notion of coNP-kernelization, could be considered folklore. Since the material never appeared in print, and has consequences for our discussion of obstruction sets, we present it here. A \emph{coNP-kernelization algorithm} (or \emph{coNP-kernel}) of size~$f \colon \mathbb{N} \to \mathbb{N}$ for a parameterized problem~$\Q$ is a polynomial-time nondeterministic Turing machine that, on input~$(x,k) \in \Sigma^* \times \mathbb{N}$, outputs an instance~$(x',k') \in \Sigma^* \times \mathbb{N}$ of size at most~$f(k)$ on each computation path, such that: (i) if~$(x,k) \in \Q$ then all computation paths output \yes-instances, and (ii) if~$(x,k) \not \in \Q$ then at least one computation path outputs a \no-instance.

\begin{lemma}[$\bigstar$] \label{lemma:efficient:po:yields:coNP:kernel}
Let~$\Q \subseteq \Sigma^* \times \mathbb{N}$ be a parameterized problem. If there is a polynomial~$p \colon \mathbb{N} \to \mathbb{N}$ and an efficiently generated quasi-order~$\preceq$ such that:
\begin{enumerate}[label=\alph*.,ref=(\alph*),topsep=2pt, partopsep=2pt] 
	\item $\Q$ is a lower ideal under~$\preceq$, and\label{coNP:kernel:lowerideal}
	\item for any~$(x,k) \not \in \Q$ there is an \emph{obstruction}~$(x',k') \not \in \Q$ of size at most~$p(k)$ with~$(x',k') \preceq (x,k)$,\label{coNP:kernel:obstruction}
\end{enumerate}
then~$\Q$ has a coNP-kernel of size~$p(k) + \Oh(1)$.
\end{lemma}

The coNP-kernel is built as follows: on input~$(x,k)$, generate the elements preceding it. If a generated element has size at most~$p(k)$ then output it, otherwise output a constant-size \yes-instance as the result of the computation path. The following theorem follows directly from Lemma~\ref{lemma:efficient:po:yields:coNP:kernel} together with the co-nondeterministic variant of Theorem~\ref{theorem:crossCompositionNoKernel} (see~\cite[Theorem 2]{KratschPRR12}).

\begin{theorem} \label{theorem:crosscomposition:obstruction:bounds}
Let~$\L$ be a language that is NP-hard under Karp reductions and that \orsymb-cross-composes into a parameterized problem~$\Q \subseteq \Sigma^* \times \mathbb{N}$. Assuming \ncontainment there is no efficiently generated quasi-order~$\preceq$ on~$\Sigma^* \times \mathbb{N}$ and polynomial~$p \colon \mathbb{N} \to \mathbb{N}$ such that:
\begin{itemize}[topsep=2pt, partopsep=2pt] 
	\item $\Q$ is a lower ideal under~$\preceq$.
	\item for any~$(x,k) \not \in \Q$ there is an \emph{obstruction}~$(x',k') \not \in \Q$ of size at most~$p(k)$ with~$(x',k') \preceq (x,k)$.
\end{itemize}
\end{theorem}

Theorem~\ref{theorem:crosscomposition:obstruction:bounds} shows that an \orsymb-cross-composition of an NP-hard set into~$\Q$ makes it unlikely that~$\Q$ admits an efficiently generated quasi-order on its instances that characterizes the problem by obstructions of polynomial size. The strength of the theorem comes from the fact that it excludes the existence of \emph{efficiently generated} quasi-orders. As a \emph{polynomial-time} quasi-order that characterizes~$\Q$ by polynomial-size obstructions places~$\Q$ in coNP, no NP-complete problem is characterized by polynomial-size obstructions under a polynomial-time quasi-order, unless \npeqconp.

Applying Theorem~\ref{theorem:crosscomposition:obstruction:bounds} to \kPathwidth, we obtain some interesting information about the properties of the pathwidth measure. While it was already known that the minor-minimal obstructions to pathwidth~$k$ can have size exponential in~$k$, Theorem~\ref{theorem:crosscomposition:obstruction:bounds} shows that \emph{any} efficiently generated quasi-order under which the \yes-instances are closed, must have superpolynomial size obstructions. As many natural quasi-orders on graphs are efficiently generated, this shows that a nice characterization of pathwidth in terms of polynomial-size obstructions is unlikely to exist, for \emph{any} efficiently generated quasi-order.

\section{Conclusion} \label{section:conclusion}
The thesis underlying this paper is that the sizes of problem kernels and the sizes of obstructions in a quasi-order are intimately related, and should be studied together. We gave a general characterization of FPT in terms of problems admitting efficiently testable quasi-orders that characterize \no-instances by obstructions of bounded size. In Sections~\ref{section:or:composition:for:pw} and~\ref{section:crosscomposition:obstruction:bounds} we showed how properties of obstruction sets can be used to derive kernelization bounds, and vice versa. There are various other examples of the strong connection between kernel sizes and obstruction sizes in the literature. We briefly discuss three of them.

1) Obstructions to list-colorability played a crucial role in the analysis of kernels for structural parameterizations of \qColoring by Jansen and Kratsch~\cite{JansenK11b}. They proved that the existence of polynomial kernels for \qColoring, parameterized by a vertex modulator to a graph class~$\F$, is determined by the existence of a bound on the size of obstructions to $q$-list-colorability of graphs in~$\F$.

2) Fomin et al.~\cite{FominLMS12} studied the \FDeletion problem. It asks for a fixed, finite family~$\F$, a graph~$G$, and an integer~$k$, whether~$k$ vertices can be removed from~$G$ to ensure that the remainder does not contain a graph in~$\F$ as a minor. For any~$\F$, the \yes-instances of parameter value at most~$k$ form a minor ideal~$\G_{\F,k}$~\cite[Theorem 6]{FellowsL88a}. Fomin et al.~proved that \FDeletion admits a polynomial kernel for every family~$\F$ that contains a planar graph. A byproduct of their kernel shows~\cite[Theorem 3]{FominLMS12} that for every such~$\F$, there is a polynomial~$p$, such that~$\G_{\F,k}$ is characterized by obstructions of size~$p(k)$.

3) The kernelization lower bound for \kRamsey given by Kratsch~\cite{Kratsch12}, is similar in spirit to Theorem~\ref{theorem:crossCompositionNoKernel}: it composes a sequence of instances of an NP-hard problem by embedding them in a larger host graph, whose size is superpolynomial with respect to the associated parameter value. The employed host graph is related to the Tur\'{a}n graph, which is extremal in Ramsey-settings. 

We conclude with two directions for future research. Can the contrapositive of Corollary~\ref{corollary:kernel:yields:small:obstructions} be used to give kernel lower bounds? Are there problems, whose kernelization complexity is still unknown, for which the nonexistence of an \orsymb-cross-composition can be proven by the contrapositive of Theorem~\ref{theorem:crosscomposition:obstruction:bounds}? We expect a further investigation of the interplay between kernels and obstructions to yield interesting insights into the structure of hard problems.

\textbf{Acknowledgments.} We thank Stefan Kratsch and Rudolf Fleischer for insightful discussions, and Hans L.\, Bodlaender for a simple proof of Lemma~\ref{lemma:improvement:npcomplete}.

\putbib[../Paper]
\end{bibunit}

\vfill
\newpage

\appendix

\begin{bibunit}[alpha]

\section{Proofs for Section \ref{section:characterization}}
\subsection{Proof of Lemma \ref{lemma:obstructions:donot:imply:kernel}}

\begin{relemma}{lemma:obstructions:donot:imply:kernel}
\ThreeColoringByComponentSize does not admit a polynomial kernel unless \containment, but there is a polynomial-time quasi-order on its instances that satisfies statement~\ref{characterization:quasiorder} of Theorem~\ref{theorem:fptCharacterization} with~$f(k) \in \Oh(k^2)$.
\end{relemma}

\begin{proof}
The classical variant of the problem is NP-complete~\cite[GT4]{GareyJ79}. For the kernelization lower bound we use the framework by Bodlaender et al.~\cite{BodlaenderDFH09}. As the disjoint union of a series of graphs is $3$-colorable if and only if all individual graphs are, \ThreeColoringByComponentSize is \andsymb-compositional. Using the mentioned framework~\cite[Lemma 7]{BodlaenderDFH09} and a recent result of Drucker~\cite{Drucker12}, this implies that \ThreeColoringByComponentSize does not admit a polynomial kernel unless \containment. 

We proceed by giving a polynomial quasi-order. Let~$(x_N, k_N)$ be a constant-size \no-instance, for example the graph~$K_4$ with a parameter value of~$4$. Now consider the following polynomial-time quasi-order on instances $(x,k) \in \{0,1\} \times \mathbb{N}$ of this problem, assuming that an instance is encoded by giving the adjacency matrix of the graph in row-major order. Say that~$(x',k') \preceq (x,k)$ if one of the following holds:
\begin{itemize}
	\item $(x',k') = (x,k)$.
	\item $(x',k')$ and~$(x,k)$ are well-formed instances, $k = k'$, and matrix~$x'$ can be formed by restricting the adjacency matrix~$x$ to the entries corresponding to vertices of one of the connected components of the graph encoded by~$x$.
	\item $(x',k') = (x_N, k_N)$ and the instance~$(x,k)$ is not well-formed, because~$x$ is not an adjacency matrix or because the corresponding graph has a connected component of more than~$k$ vertices.
\end{itemize}
It is easy to see that~$\preceq$ is a quasi-order on instances of \ThreeColoringByComponentSize. As the explicit encoding by adjacency matrices avoids the need for graph isomorphism tests when comparing two instances, $\preceq$~can be decided in polynomial time: test for each connected component of the graph encoded by~$x$ whether the relevant restriction yields~$x'$. 

If~$(x',k') \preceq (x,k)$ for well-formed instances~$(x,k)$ and~$(x',k')$, then the graph encoded by~$x'$ is an induced subgraph of the graph encoded by~$x$. Hence the \yes-instances are closed under~$\preceq$, as three-colorability is a hereditary property. Therefore~$\preceq$ satisfies \ref{characterization:lowerideal} of Theorem~\ref{theorem:fptCharacterization}.

To see that \ref{characterization:obstruction} is also satisfied, consider a \no-instance~$(x,k)$. If it is not well-formed, then~$(x_N, k_N) \preceq (x,k)$ which has constant size. If~$(x,k)$ is well-formed then~$x$ encodes a graph of chromatic number at least four, in which every connected component has at most~$k$ vertices. Hence it has a connected component~$C$ on at most~$k$ vertices that is not $3$-colorable. Let~$x'$ be the restriction of matrix~$x$ to the vertices in~$C$. Then~$(x',k)$ is a \no-instance that precedes~$(x,k)$. As~$C$ has at most~$k$ vertices,~$|x'| \leq k^2$ and therefore~$|(x',k)| \in \Oh(k^2)$. Hence there is indeed a quadratic-size obstruction preceding each \no-instance, which concludes the proof.
\qed
\end{proof}

\section{Proofs for Section \ref{section:or:composition:for:pw}} \label{appendix:crosscomposition}

\subsection{Definition of OR-Cross-Composition}
\begin{definition}[{\cite{BodlaenderJK11}}] \label{definition:eqvr}
An equivalence relation~\eqvr on $\Sigma^*$ is called a \emph{polynomial equivalence relation} if the following two conditions hold:
\begin{enumerate}
	\item There is an algorithm that, given two strings~$x,y \in \Sigma^*$, decides whether~$x$ and~$y$ belong to the same equivalence class in~$(|x| + |y|)^{\Oh(1)}$ time.
	\item For any finite set~$S \subseteq \Sigma^*$, the equivalence relation~$\eqvr$ partitions the elements of~$S$ into at most~$(\max _{x \in S} |x|)^{\Oh(1)}$ classes.
\end{enumerate}
\end{definition}
\begin{definition}[{\cite{BodlaenderJK11}}] \label{definition:crossComposition}
Let~$\L \subseteq \Sigma^*$ be a set and let~$\Q \subseteq \Sigma^* \times \mathbb{N}$ be a parameterized problem. We say that~$\L$ \emph{\orsymb-cross-composes} into~$\Q$ if there is a polynomial equivalence relation~$\eqvr$ and an algorithm that, given~$t$ strings~$x_1, x_2, \ldots, x_t$ belonging to the same equivalence class of~$\eqvr$, computes an instance~$(x^*,k^*) \in \Sigma^* \times \mathbb{N}$ in time polynomial in~$\sum _{i=1}^t |x_i|$ such that:
\begin{enumerate}
	\item~$(x^*, k^*) \in \Q \Leftrightarrow x_i \in \L$ for some~$i \in [t]$,
	\item~$k^*$ is bounded by a polynomial in~$\max _{i=1}^t |x_i|+\log t$.
\end{enumerate}
\end{definition}

\subsection{Preliminaries of Path Decompositions}
We use~$K_n$ to denote the complete graph on~$n$ vertices. The disjoint union of~$t$ copies of a graph~$G$ is~$t \cdot G$. The \emph{join} of two graphs~$G_1$ and~$G_2$ is the graph~$G_1 \otimes G_2$ on the vertex set~$V(G_1) \dotcup \linebreak[1] V(G_2)$ and edge set~$E(G_1) \dotcup E(G_2) \dotcup \linebreak[0] \{ \{x,y\} \mid x \in V(G_1) \wedge y \in V(G_2) \}$. We need the following properties of path decompositions.

\begin{lemma}[\cite{BodlaenderM93}]\label{lemma:treewidth:join}Any two graphs~$G_1$ and~$G_2$ satisfy 
\[
\pw(G_1\otimes G_2)=\min(\pw(G_1)+|V(G_2)|,\pw(G_2)+|V(G_1)|).
\]
\end{lemma}

\begin{lemma}[\cite{BodlaenderM93}] \label{lemma:cliquecontainment} 
If~$S$ is a clique in graph~$G$, then any path decomposition of~$G$ has a bag containing~$S$.
\end{lemma}

\begin{lemma}[{Compare~\cite[Theorem 29]{Bodlaender98}}] \label{lemma:pathwidth:interval:supergraphs}
Any graph~$G$ satisfies 
\[
\pw(G) - 1 = \min \{ \omega(H) - 1 \mid \mbox{H is an interval graph and~$G \subseteq H$\}}.
\]
\end{lemma}

\subsection{Proof of Lemma \ref{lemma:improvement:npcomplete}}

\begin{relemma}{lemma:improvement:npcomplete}
\PathwidthImprovement is NP-complete.
\end{relemma}

\begin{proof}
Membership in NP is trivial. To establish completeness we reduce from the NP-complete~\cite{ArnborgCP87} \Pathwidth problem. An instance of \Pathwidth consists of a graph~$G$ and integer~$k \geq 0$. The question is whether~$\pw(G) \leq k$. Let~$n := |V(G)|$. If~$k \geq n - 1$ then the answer is trivially \yes and we output a constant-size \yes-instance. In the remainder let~$i := n - k - 1$, which is at least one. Construct the graph~$G' := G \otimes (i \cdot K_1)$, i.e., the join of~$G$ and an independent set on~$i$ vertices. Construct a path decomposition~$\P'$ for~$G'$ by making~$i$ consecutive bags; each bag contains~$V(G)$ and exactly one vertex from the independent set~$i \cdot K_1$. It is easy to verify that this constitutes a valid path decomposition of~$G'$. 

Let~$k' := n + 1$. As each bag of~$\P'$ contains~$n + 1$ vertices, the width of~$\P'$ is~$k' - 1$. Hence~$(G', k', \P')$ is a valid instance of \PathwidthImprovement, which asks whether~$G'$ has pathwidth at most~$k' - 2$. Using Lemma~\ref{lemma:treewidth:join} together with the fact that the edgeless graph~$i \cdot K_1$ has pathwidth zero, it easily follows that~$G$ has pathwidth at most~$k$ if and only if~$G'$ has pathwidth at most~$k'-2$. As the instance~$(G', k', \P')$ of \PathwidthImprovement can be constructed in polynomial time, this concludes the proof.
\qed
\end{proof}

\subsection{Proof of Lemma \ref{lemma:ternarytrees:obstructions}}

\begin{relemma}{lemma:ternarytrees:obstructions}
For~$k \in \mathbb{N}_0$ the graph~$\tobs^k$ is a minor-minimal obstruction to pathwidth~$k$.
\end{relemma}

\begin{proof}
It is known that if~$G_1, G_2, G_3$ are acyclic, connected, minor-minimal obstructions to pathwidth~$k$, then any graph that can be formed from the disjoint union~$G_1 \dotcup G_2 \dotcup G_3$ by adding a new vertex that is made adjacent to exactly one vertex in each of the components~$G_1, G_2, G_3$, is a minor-minimal obstruction to pathwidth~$k+1$. This was proven independently by Kinnersley~\cite[Theorem 4.3]{Kinnersley92} and Takahashi et al.~\cite[Theorem 2.5]{TakahashiUK94}. Observe that~$K_2 = \tobs^0$ is the unique acyclic, connected, minor-minimal obstruction to pathwidth zero. As~$\tobs^k$ can be built by connecting a new vertex to the roots of three copies of~$\tobs^{k-1}$, the lemma follows from the cited results by induction.
\qed
\end{proof}

\subsection{Proof of Lemma \ref{lemma:pathwidth:inflation}}
\begin{relemma}{lemma:pathwidth:inflation}
For any graph~$G$ and~$k \in \mathbb{N}: \pw(G \diamond k) + 1 = k \cdot (\pw(G) + 1)$.
\end{relemma}
\begin{proof}
The lemma can be considered folklore; we give a proof for completeness. Given a path decomposition for~$G$, one may replace each occurrence of a vertex~$v$ by the copies~$v_1, \ldots, v_k$ to obtain a path decomposition of~$G \diamond k$. The size of each bag is multiplied by~$k$. Hence a path decomposition of width~$\pw(G)$, whose bags have size at most~$\pw(G) + 1$, is transformed into a decomposition of~$G \diamond k$ of maximum bag size at most~$k \cdot (\pw(G) + 1)$. The width of such a decomposition is at most~$k \cdot (\pw(G) + 1) - 1$, implying that~$\pw(G \diamond k) + 1 \leq k \cdot (\pw(G) + 1)$.

For the other direction, we first show that there is a minimum-width decomposition of~$G \diamond k$ where for each~$v \in V(G)$, each bag of the decomposition either contains all of~$v_1, \ldots, v_k$ or none of them. Consider an arbitrary path decomposition~$\P' = (\X'_1, \ldots, \X'_r)$ of~$G \diamond k$. Define a decomposition~$\P'' = (\X''_1, \ldots, \X''_r)$ as follows: for each vertex~$v \in V(G)$, if~$\{v_1, \ldots, v_k\} \subseteq \X'_i$ then~$\{v_1, \ldots, v_k\}$ is added to~$\X''_i$. Then~$\P''$ trivially has the desired form, while its width is not greater than the width of~$\P'$. It remains to prove that~$\P''$ is a valid path decomposition.

For each vertex~$v \in V(G)$, the set~$v_1, \ldots, v_k$ is a clique in~$G'$. By Lemma~\ref{lemma:cliquecontainment} it follows that~$\P'$ has a bag~$\X'_i$ containing all vertices~$v_1, \ldots, v_k$. It follows that~$\X''_i$ also contains those vertices. It is easy to verify that the convexity property is not violated. It remains to show that all edges are represented. For~$v \in V(G)$ and~$i,j \in [k]$, the edge~$\{v_i, v_j\}$ is represented in a bag that contains the clique~$\{v_1, \ldots, v_k\}$. It remains to show that for each edge~$\{u,v\}$ of~$G$, all edges between copies of~$u$ and~$v$ are represented in~$\P''$. Observe the following. For each choice of~$x,y \in \{u,v\}$ (possibly~$x=y$) and~$i,j \in [k]$ (possibly~$i,j$), either~$x_i = y_j$ or there is an edge~$\{x_i, y_j\}$ in~$G \diamond k$ that decomposition~$\P'$ represents in some bag. Hence in~$\P'$, the interval of bags that contain~$x_i$, intersects the interval of bags that contain~$y_j$, for all valid choices of~$x,y$ and~$i,j$. By the Helly property\footnote{A family of sets~\F has the \emph{Helly property} if every subfamily~$\F' \subseteq \F$ with~$\forall S,T \in \F': S \cap T \neq \emptyset$ satisfies~$\bigcap _{S \in \F'} S \neq \emptyset$. It is well-known that the set of intervals on the real line has the Helly property.} for closed intervals on the real line, it follows that the common intersection of all these intervals is nonempty. Hence there is a bag in~$\P'$ containing~$\{u_1, \ldots, u_k, v_1, \ldots, v_k\}$. The corresponding bag in~$\P''$ realizes the edges between the copies of~$u$ and the copies of~$v$, which proves that~$\P''$ is indeed a valid path decomposition of~$G \diamond k$.

Now we can prove that~$k \cdot (\pw(G) + 1) \leq \pw(G \diamond k) + 1$. Consider a minimum-width path decomposition~$\P'$ of~$G \diamond k$ where, for each~$v \in V(G)$, each bag contains either all copies of~$v$ or none. Removing all copies with index two or higher from the decomposition results in a path decomposition of a graph isomorphic to~$G$. Since the number of copies of a vertex in each bag drops by a factor~$k$, the maximum size of a bag is decreased by a factor~$k$. As the width of a decomposition is one less than the maximum bag size, this implies that~$k \cdot (\pw(G) + 1) \leq \pw(G \diamond k) + 1$.
\qed
\end{proof}

\subsection{Proof of Theorem \ref{theorem:pathwidth:crosscomposition}}
\begin{retheorem}{theorem:pathwidth:crosscomposition}
The \PathwidthImprovement problem \orsymb-cross-composes into \kPathwidth.
\end{retheorem}
\begin{proof}
Choose an equivalence relation~\eqvr under which any two strings that do not encode valid instances of \PathwidthImprovement are equivalent. Let two well-formed instances be equivalent if they ask for the same pathwidth bound. It is easy to verify that this is a polynomial equivalence relation. We show how to cross-compose instances of \PathwidthImprovement that are equivalent under~\eqvr. If the instances are malformed, then we may output a constant-size \no-instance. In the remainder, we may therefore assume that the input consists of a sequence~$(G_1, k_1, \P^1), \ldots, \linebreak[0] (G_t, k_t, \P^t)$ such that~$k_1 = \ldots = k_t = k$. The instance with index~$i \in [t]$ asks whether~$G_i$ has pathwidth at most~$k-2$, and supplies a path decomposition~$\P^i$ of~$G_i$ having width at most~$k-1$. By duplicating some inputs, we can enforce that~$t$ is a power of three and~$t \geq 3$; this does not change \orsymb of the answers to the input instances, while at most increasing the input size by a factor three. So let~$t$ be equal to~$3^s$ for some~$s \in \mathbb{N}$.

The construction is based on the minor-minimal obstruction~$\tobs^s$. Label the $3^s = t$ leaves of~$\tobs^s$ as~$x^1, \ldots, x^t$, and let~$y^1, \ldots, y^t$ be the parents of those leaves. As~$s \geq 1$ each vertex~$y_i$ has degree exactly two in~$\tobs^s$. We cross-compose the instances into a single graph~$G'$. It is obtained by inflating~$\tobs^s$ by a factor~$k$ and replacing each $k$-vertex clique containing the copies of a leaf~$x_i$ by the graph~$G_i$. More formally, we obtain~$G'$ as follows.
\begin{itemize}
	\item Initialize~$G'$ as the inflation~$\tobs^s \diamond k$. For each leaf~$x^i$ the copies created by the inflation form a clique of size~$k$ on vertices~$x^i_1, \ldots, x^i_k$.
	\item For each~$i \in [t]$, remove the vertices~$x^i_1, \ldots, x^i_k$ from~$G'$ and replace them by a copy of the graph~$G_i$. Make all vertices of~$G_i$ adjacent to the copies of the parent of~$x^i$, i.e., to the vertices~$y^i_1, \ldots, y^i_k$.
\end{itemize}
Refer to Fig.~\ref{img:crosscomposition} for an example. Let~$k' := k (s+2) - 2 \in \Oh(n \cdot \log t)$.

\begin{claim}
$\pw(G') \leq k'$ if and only if there is an~$i \in [t]$ such that~$\pw(G_i) \leq k - 2$.
\end{claim}
\begin{proof}
($\Rightarrow$) Assume that~$G'$ has pathwidth at most~$k'$. By Lemma~\ref{lemma:pathwidth:interval:supergraphs} there is an interval graph~$H'$ that is a supergraph of~$G'$, such that~$\omega(H') \leq k' +  1$. We show that there is an index~$i^* \in [t]$ such that the induced subgraph~$H'[V(G_{i^*})]$ has clique number at most~$k-1$. 

To see this, suppose that no such index exists. For each~$i \in [t]$ let~$S_i$ be a clique of size~$k$ in~$H'[V(G_i)]$. Consider the graph~$H''$ obtained from~$H'$ by removing the vertices~$V(G_i) \setminus S_i$ for each~$i \in [t]$; from each graph that was plugged into~$\tobs^i \diamond k$, we only leave a single clique. As interval graphs are hereditary,~$H''$ is an interval graph. Since~$H'' \subseteq H'$ we have~$\omega(H'') \leq \omega(H') \leq k'+1$. Now observe that~$H''$ is a supergraph of~$\tobs^i \diamond k$: when constructing~$G'$ we replaced the inflated size-$k$ cliques corresponding to the leaves of~$\tobs^i \diamond k$ by the input graphs, joined to all copies of the parent vertices of that leaf, but we reduce the vertex sets of the plugged-in graphs~$G_i$ back to size-$k$ cliques. But then~$H''$ is an interval supergraph of~$\tobs^i \diamond k$ of clique number at most~$k'+1 = k (s+2) - 1$, implying by Lemma~\ref{lemma:pathwidth:interval:supergraphs} that~$\pw(\tobs^i \diamond k) \leq k (s+2) - 2$. But this contradicts the fact that~$\pw(\tobs^i \diamond k)$ has pathwidth at least~$k (s+2) - 1$, which follows from Lemma~\ref{lemma:ternarytrees:obstructions} and Lemma~\ref{lemma:pathwidth:inflation}.

Thus there must indeed by an index~$i^* \in [t]$ such that~$H'[V(G_{i^*})]$ has clique number at most~$k-1$. As that graph is an interval supergraph of~$G_{i^*}$, this shows by Lemma~\ref{lemma:pathwidth:interval:supergraphs} that~$G_{i^*}$ has pathwidth at most~$k-2$, proving this direction of the equivalence.

($\Leftarrow$) For the reverse direction, suppose that~$i^* \in [t]$ such that~$\pw(G_{i^*}) \leq k - 2$. Let~$\P^*$ be a path decomposition of~$G_{i^*}$ whose bags have size at most~$k-1$. Since~$\tobs^s$ is a minor-minimal obstruction to pathwidth~$s$, its proper minor~$\tobs^s - x_{i^*}$ has pathwidth at most~$s$. By Lemma~\ref{lemma:pathwidth:inflation} this implies that the inflation~$(\tobs^s - x_{i^*}) \diamond k$ has pathwidth at most~$k (s+1) - 1$. Consider a path decomposition~$\P'$ of~$(\tobs^s - x_{i^*}) \diamond k$ whose maximum bag size is bounded by~$k (s+1)$. We will  transform~$\P'$ into a path decomposition of~$G'$ of width at most~$k'$.

We first show how to incorporate the graphs~$G_i$ into~$\P'$ for~$i \in [t] \setminus \{i^*\}$. For each such~$i$, do the following. Consider the leaf~$x^i$ of the obstruction~$\tobs^s$, with the corresponding parent~$y^i$. By the existence of the edge~$\{x^i, y^i\}$ in the obstruction, the copies of~$x^i$ and~$y^i$ together form a clique in the inflated graph~$(\tobs^s - x_{i^*}) \diamond k$. Hence by Lemma~\ref{lemma:cliquecontainment} there is a bag~$\X'_j$ of the decomposition~$\P'$ that contains all copies of~$x^i$ and all copies of~$y^i$. We now splice the width-$(k-1)$ path decomposition~$\P^i = (\X^i_1, \ldots, \X^i_r)$ of~$G_i$ into the decomposition~$\P'$, as follows. Remove all copies of vertex~$x^i$ from the decomposition; this decreases the size of bag~$\X'_j$ by~$k$. Make~$r-1$ copies of bag~$\X'_j$ and insert them just after~$\X'_j$. Now add the contents of bag~$\X^i_1$ to~$\X'_j$, add~$\X^i_2$ to~$\X'_{j+1}$, and so on. Each of the new bags is thus obtained from~$\X'_j$ by removing the~$k$ copies of~$x^i$ and replacing them by the contents of a bag in the decomposition~$\P^i$. The latter decomposition has bags of size at most~$k$. Hence after the replacement, the maximum bag size is still at most~$k (s+1)$.

The updated decomposition correctly represents the edges in the graph~$G_i$, together with its edges to the copies of~$y^i$: 
\begin{itemize}
	\item The edges of~$G_i$ are represented in the decomposition~$\P^i$ that we spliced in.
	\item The edges between~$V(G_i)$ and the copies of~$y^i$ are represented because~$\X'_j$ contains all copies of~$y^i$, while vertices of~$V(G_i)$ only occur in bags that are copies of~$\X'_j$.
\end{itemize}
By independently splicing in the path decompositions of all input graphs except for~$G_{i^*}$, we obtain a decomposition of the graph~$G' - V(G_{i^*})$ whose bags have size at most~$k (s+1)$.

It remains to incorporate the vertices of~$G_{i^*}$ in the decomposition. Observe that the copies of~$y^{i^*}$, which are adjacent to all vertices of~$V(G_{i^*})$ in~$G'$, form a clique in~$G' - V(G_{i^*})$. By Lemma~\ref{lemma:cliquecontainment} the decomposition resulting from the previous step therefore has a bag~$\X'_{j^*}$ containing all copies of~$y^{i^*}$. Recall that~$\P^*$ is a path decomposition of~$G_{i^*}$ whose bags have size at most~$k-1$. Let~$\P^* = (\X^*_1, \ldots, \X^*_{r^*})$. Make~$r^*-1$ copies of bag~$\X'_{j^*}$, inserting them right after~$\X'_{j^*}$. Add the contents of~$\X^*_1$ to~$\X'_{j^*}$, add~$\X^*_2$ to~$\X'_{j^*+1}$, and so on, until the entire decomposition~$\P^*$ is incorporated into~$\P'$. We add bags of size at most~$k-1$ to a decomposition whose bags had size at most~$k (s+1)$. Hence the bags in the final decomposition~$\P'$ have size at most~$k (s+2) - 1$. It follows that~$\P'$ has width at most~$k (s+2) - 2$. To see that~$\P'$ is a valid path decomposition of~$G'$, observe that all edges of~$G'[V(G_{i^*})]$ are represented because we spliced in a decomposition of~$G'$. The edges between~$V(G_{i^*})$ and the copies of the parent~$y^{i^*}$ are represented because all copies of~$y^{i^*}$ are present in all bags that we splice~$\P^*$ into. Thus~$G'$ has pathwidth at most~$k (s+2) - 2 = k'$.
\claimqed
\end{proof}
The claim shows that the cross-composed instance~$(G', k')$ acts as the \orsymb of the inputs. As~$G'$ can be built in polynomial time and~$k'$ is bounded by a polynomial in the size of the largest input instance plus~$\log t$, this concludes the proof.
\qed
\end{proof}

\section{Proofs for Section \ref{section:crosscomposition:obstruction:bounds}}
\subsection{Proof of Lemma \ref{lemma:efficient:po:yields:coNP:kernel}}

\begin{relemma}{lemma:efficient:po:yields:coNP:kernel}
Let~$\Q \subseteq \Sigma^* \times \mathbb{N}$ be a parameterized problem. If there is a polynomial~$p \colon \mathbb{N} \to \mathbb{N}$ and an efficiently generated quasi-order~$\preceq$ such that:
\begin{enumerate}[label=\alph*.,ref=(\alph{enumii})]
	\item $\Q$ is a lower ideal under~$\preceq$, and
	\item for any~$(x,k) \not \in \Q$ there is an \emph{obstruction}~$(x',k') \not \in \Q$ of size at most~$p(k)$ with~$(x',k') \preceq (x,k)$,
\end{enumerate}
then~$\Q$ has a coNP-kernel of size~$p(k) + \Oh(1)$.
\end{relemma}

\begin{proof}
If~$\Q$ only has \no-instances, then the algorithm that always outputs~$(\epsilon, 1)$ trivially satisfies the requirements ($\epsilon$ denotes the empty string). Otherwise, let~$(x_Y, k_Y)$ be a \yes-instance. On input~$(x,k)$, the coNP-kernelization proceeds as follows. It simulates the NDTM that efficiently generates~$\preceq$ on~$(x,k)$. In each computation path, after the generating algorithm outputs an element~$(x',k') \preceq (x,k)$, the kernel tests whether~$|(x',k')| \leq p(k)$. If this is the case then~$(x',k')$ is used as output; otherwise~$(x_Y, k_Y)$ is used as output. 

Each computation path of the nondeterministic procedure terminates in polynomial time (as~$\preceq$ is efficiently generated) and outputs an instance of size at most~$\max(p(k), |(x_Y, k_Y)|) \leq p(k) + \Oh(1)$. To see that the procedure satisfies the nondeterministic form of correctness required from a coNP-kernel, consider the behavior on some input instance~$(x,k)$. If~$(x,k)$ is a \yes-instance, then the outputs either precede~$(x,k)$ (implying they are \yes-instances by~\ref{coNP:kernel:lowerideal}) or equal the trivial \yes-instance~$(x_Y, k_Y)$. Hence all computation paths output small \yes-instances.

Now assume that~$(x,k)$ is a \no-instance. By~\ref{coNP:kernel:obstruction} there is a \no-instance~$(x',k')$ of size at most~$p(k)$ that precedes~$(x,k)$. By definition of an efficiently generated quasi-order, at least one computation path generates~$(x',k')$. As it is sufficiently small, that instance is used as the output to the coNP-kernel. Hence at least one computation path outputs a \no-instance, as required.
\qed
\end{proof}

\section{Bibliography for the Appendix}
In this section we list the bibliographic information for items which were referenced in the appendix, but not in the main text.

\putbib[../Paper]
\end{bibunit}

\end{document}